\documentclass[11pt,letterpaper]{article}

\usepackage[margin=1in]{geometry}
\usepackage{setspace}

\usepackage[linesnumbered,ruled,vlined]{algorithm2e}

\usepackage{algorithmicx}
\usepackage{algpseudocode}
\usepackage{aliascnt}
\usepackage{todonotes}
\usepackage{pdfpages}
\usepackage[margin=1in]{geometry}
\usepackage{setspace}
\usepackage{parskip}
\usepackage[colorlinks=true,
            linkcolor=blue,
            citecolor=green,
            urlcolor=green]{hyperref}

\usepackage{amsthm}
\usepackage{amsmath}
\DeclareMathOperator{\poly}{poly}

\usepackage{times}
\usepackage{microtype}
\usepackage{enumitem}
\usepackage{titlesec}
\usepackage{soul} 

\usepackage{amsmath, amssymb}
\usepackage{bbm}

\AtBeginDocument{%
  }

\makeatletter
\newcommand{\prob}[1]{%
  \@ifnextchar\bgroup{\prob@two{#1}}{\prob@one{#1}}%
}
\newcommand{\prob@one}[1]{\mathbb{P}\left[#1\right]}
\newcommand{\prob@two}[2]{\mathbb{P}_{#1}\left[#2\right]}
\makeatother
\makeatletter
\renewcommand{\exp}[1]{%
  \@ifnextchar\bgroup{\exp@two{#1}}{\exp@one{#1}}%
}
\newcommand{\exp@one}[1]{\mathbb{E}\left[#1\right]}
\newcommand{\exp@two}[2]{\mathbb{E}_{#1}\left[#2\right]}
\makeatother

\newcommand{\eps}{\varepsilon}

\newtheorem{theorem}{Theorem}[section]

\newaliascnt{lemma}{theorem}
\newtheorem{lemma}[lemma]{Lemma}
\aliascntresetthe{lemma}

\newaliascnt{observation}{theorem}
\newtheorem{observation}[observation]{Observation}
\aliascntresetthe{observation}

\newaliascnt{claim}{theorem}
\newtheorem{claim}[claim]{Claim}
\aliascntresetthe{claim}

\newaliascnt{definition}{theorem}
\newtheorem{definition}[definition]{Definition}
\aliascntresetthe{definition}

\newaliascnt{remark}{theorem}
\newtheorem{remark}[remark]{Remark}
\aliascntresetthe{remark}
\usepackage{tikz}
\usepackage{cleveref}

\crefname{theorem}{theorem}{theorems}
\Crefname{theorem}{Theorem}{Theorems}

\crefname{lemma}{lemma}{lemmas}
\Crefname{lemma}{Lemma}{Lemmas}

\crefname{observation}{observation}{observations}
\Crefname{observation}{Observation}{Observations}

\crefname{claim}{claim}{claims}
\Crefname{claim}{Claim}{Claims}

\date{}

\title{Faster Sublinear Maximal Independent Set Size}

\author{
Peter Kiss\thanks{This research was funded in whole or in part by the Austrian Science Fund
(FWF) 10.55776/ESP6088024, peter.kiss@univie.ac.at}$^{1}$,
Arash Kooroshnezhad$^{2}$ \thanks{kouroshnejad04@gmail.com} \\
\\
$^{1}$University of Vienna 
$^{2}$University of Warwick
}

\begin{document}
\maketitle
\begin{abstract}

We give a sublinear-time algorithm for estimating the size of a maximal independent set in a graph using adjacency-query access with expected running time $\tilde{O}(n^{1+1/3})$, improving over the previous $\tilde{O}(n^{1+1/2})$ bound of \cite{MahabadiRTV25}. As a consequence of a reduction of \cite{MahabadiRTV25}, this also improves the running time for sublinear metric Steiner forest. We further show that a bi-criteria estimate of the $k$-center objective in general metrics can be obtained via a reduction to maximal independent set size estimation.

\end{abstract}
\section{Introduction}

\label{sec:intro}

A maximal independent set is one of the most fundamental combinatorial objects in graph algorithms. Given a graph $G=(V,E)$, an independent set $I \subseteq V$ is \emph{maximal} if no vertex can be added to $I$ while preserving independence. Equivalently, $I$ is an independent dominating set: every vertex outside $I$ has a neighbor in $I$. This dual nature makes maximal independent sets central to distributed algorithms, approximation algorithms, symmetry breaking, and graph sparsification.

Computing a maximal independent set is a crucial subroutine in a wide range of graph algorithms. It appears prominently in distributed and parallel symmetry breaking \cite{BarenboimEPS16,kcenterbase}, where it underlies algorithms for coloring \cite{PanconesiS92}, ruling sets \cite{KothapalliP12,BarenboimEPS16}, network decompositions \cite{PanconesiS92, AwerbuchGLP89,RozhonG20}, and maximal matching \cite{BalliuBHORS21}. It is also used in approximation algorithms, for example through its interpretation as an independent dominating set \cite{kcenterbase,JainV01}, and in graph sparsification and clustering procedures \cite{BaswanaS07,AhmedBSHJKS20} where one needs to select a well-separated yet covering set of vertices.

Recently, there has been growing interest in algorithms whose running time is sublinear in the size of the input. This line of work has led to sublinear-time algorithms for fundamental graph and metric problems, such as matching \cite{YoshidaYMIH09, KapralovMNT20, Behnezhad21, BehnezhadRR23a, BehnezhadRRS23, BhattacharyaKS23a, AzarmehrBRR25, MahabadiRT25, BehnezhadRR24, BehnezhadRR23b, BhattacharyaKS23b}, metric TSP \cite{ChenKK20, 0039KT23b,BehnezhadRRS24}, minimum spanning tree \cite{CzumajS09, ChazelleRT05}, spanners \cite{LenzenL18, ArvivCLP23,ParterRVY19} and metric Steiner tree \cite{MahabadiRTV25b, 0039KT23,MahabadiRTV25} . In the sublinear regime, however, one is often unable to construct an exact feasible solution. Instead, most work in this model focuses on estimating the size or cost of an optimal solution, using only a small number of queries to the input.

In this paper we focus on the problem of how fast the size of a maximal independent set of the input graph can be estimated within a $(1+\epsilon)$-multiplicative factor. As MIS-based subroutines have found a wide range of applications, it's natural to propose that in a sublinear implementation of an algorithm relying on finding an independent set at some point, this challenge forms a significant bottleneck. Indeed, Mahabadi, Roghani, Tarnawski and Vakilian \cite{MahabadiRTV25} have shown that one can obtain an $O(\log k)$-approximation to the metric Steiner forest problem in time proportional to how fast an independent set size estimation can be obtained.

The maximal independent set problem has already seen significant research interest in sublinear models. Most prior work \cite{YoshidaYMIH09,NguyenO08,000122,RubinfeldTVX11,AlonRVX12, ReingoldV16, LeviRY17,GhaffariU19,EvenMR18,LeviM17} has focused on the following local-computation question: given a vertex $v$, how quickly can one decide whether $v$ belongs to the MIS produced by a fixed algorithm, while inspecting only a small portion of the graph? Only recently the authors of \cite{MahabadiRTV25} have shown the first sublinear running time algorithm for estimating the independent set size globally.

A key question in sublinear-time graph algorithms is the choice of access model. Since the algorithm cannot read the entire graph, its power is determined by the queries it is allowed to make. Two standard models are the \emph{adjacency-list query} model and the \emph{adjacency-matrix query} model. In the adjacency-list model, a query $(v,i)$ returns the $i$th neighbor of vertex $v$ according to a fixed ordering of its adjacency list. In the adjacency-matrix model, a query $(u,v)$ reveals whether the edge $(u,v)$ is present in the graph.

In their seminal work, Yoshida, Yamamoto and Ito ~\cite{YoshidaYMIH09} studied local access to the maximal independent set produced by the random-order greedy algorithm. They showed that the MIS status of a uniformly random vertex can be determined in expected $O(m/n)$ time using adjacency-list queries. The authors of \cite{MahabadiRTV25} later extended this guarantee to the adjacency-matrix query model, showing that the status of a uniformly random vertex can be determined in expected $O(n)$ time.

A natural approach to estimating the size of a solution in the sublinear model is to sample vertices or edges, decide whether each sampled object participates in the output solution, and use the resulting sample average as a global estimator. The authors of \cite{MahabadiRTV25} follow this approach to obtain a $(1+\epsilon)$-approximation to the size of a maximal independent set in $\tilde{O}(n^{1+1/2})$ time. 

At a high level, their algorithm combines two complementary strategies. First, it is a folklore fact that, with adjacency-matrix access, the random-order greedy MIS can be simulated in time $O(S\cdot n)$, where $S$ denotes the size of the output independent set. Thus, when $S=o(\sqrt n)$, the MIS can be constructed explicitly in $\tilde{O}(n^{1+1/2})$ time. Second, when $S = \Omega(\sqrt{n})$, it suffices to sample $\tilde{O}(\sqrt n)$ vertices uniformly at random, decide for each sampled vertex whether it belongs to the greedy MIS, and extrapolate from the sample to estimate the total size.

Unfortunately, this black-box use of the MIS membership algorithm of ~\cite{YoshidaYMIH09} appears to encounter an inherent barrier at $\tilde{O}(n^{1+1/2})$ time. The hard regime is when the output MIS has size roughly $\sqrt n$. In this case, explicitly constructing the solution is already too expensive, while a sampling-based approach is also ineffective: a uniformly random vertex belongs to the MIS with probability only about $1/\sqrt n$, and each membership query costs $\tilde{O}(n)$ time. Thus, finding even a single MIS vertex by repeatedly sampling vertices and testing their membership requires $\Omega(n^{1+1/2})$ time.

Our main contribution is in showing an approach breaking this barrier. The guarantee of our algorithm is summarized in the following theorem.

\begin{theorem}
\label{thm:main}
There exists a sublinear-time algorithm which, given adjacency-matrix query access to an $n$-vertex graph $G$, returns a $(1+\epsilon)$-approximation to the size of the maximal independent set produced by the random-order greedy algorithm on $G$ with high probability in expected running time $\tilde{O}(\frac{n^{1+1/3}}{\epsilon^2})$.
\end{theorem}

The main technical novelty of our approach is that we do not use the membership oracle of ~\cite{YoshidaYMIH09} as a black box. Instead, we adapt its internal recursive procedure so that it can be initialized with, and guided by, a partially computed maximal independent set.

\paragraph{Applications.}
Mahabadi et al.~\cite{MahabadiRTV25} showed that the metric Steiner forest problem admits an $O(\log k)$-approximation in time proportional to the cost of estimating the size of a maximal independent set under adjacency-matrix query access. Our improved MIS-size estimator therefore accelerates their algorithm as a black box.

We also give a new application to clustering in general metrics. By combining MIS-size estimation with the folklore threshold-graph algorithm (usually attributed to \cite{kcenterbase}), we estimate, within a factor of $2+\epsilon$, the smallest radius for which the metric can be covered by at most $(1+\epsilon)k$ centers (see \Cref{sec:k-center}). For general metrics, such a bi-criteria relaxation is necessary for sublinear dependence on $nk$, since any constant-factor non-bi-criteria approximation requires $\Omega(nk)$ distance queries~\cite{BateniEFHJMW23,abs-2112-07050}. This complements the near-linear-time bi-criteria schemes known for doubling metrics~\cite{Cohen-AddadFS21}.

\paragraph{Concurrent Work.} In concurrent and independent work Mahadabi and Tarnawski \cite{mahabadi2026improvedsublinearalgorithmsmaximal} obtained a sub-linear algorithm for maximal independent set size with similar guarantees as \Cref{thm:main}. Similarly to our approach, their algorithm focuses on constructing a pre-fix of the run of the greedy MIS algorithm with respect to a random permutation. In contrast, their algorithm fixes the length of the pre-fix constructed instead of bounding the number of queries made during the pre-fix computation. In addition, they obtained an $O(\log k)$-approximate sub-linear algorithm for Steiner Forest with $\tilde{O}(k)$ running time, bypassing sub-linear MIS size estimation.

\subsection{Technical Overview}

\label{sec:tech-overview}

Our algorithm estimates the size of the independent set produced by the random-order greedy algorithm. It first computes a \emph{baseline} independent set $S_b \subseteq S$, where $S = \text{RGMIS}(G,\pi)$, by running the greedy algorithm with respect to a random permutation of the vertices $\pi$ until $\Theta(n^{1+1/3})$ adjacency-matrix queries have been made.

Since processing each vertex uses at most $|S|$ queries, if $|S| \leq n^{1/3}$, then this already computes the entire maximal independent set. Thus we may assume $|S|=\Omega(n^{1/3})$. Let $S_b(i)$ denote the $i$-th vertex added to $S_b$, and define

\begin{equation*}
\lambda = \max_{i <|S_b|} \left( \pi^{-1}(S_b(i+1)) - \pi^{-1}(S_b(i)) \right), \ \ \ \ \ \ U = V \setminus (S_b \cup N_G(S_b))
\end{equation*}

Here, $\lambda$ is the length of the longest chain of vertices not added to $S_b$, and $U$ is the set of vertices not yet dominated by $S_b$. Since the suffix of a uniformly random permutation remains uniformly random after conditioning on any prefix, one expects $|U|=\tilde{O}(n/\lambda)$ with high probability.

Moreover, since among every $\lambda$ consecutive vertices in $\pi$ at least one is added to $S_b$, and since constructing $S_b$ takes $\Omega(n^{1+1/3})$ queries with at most $|S_b|$ queries per processed vertex, we obtain
$|S_b| \geq n^{1+1/3}/(|S_b|\lambda)$, or equivalently $|S_b| \geq n^{2/3}/\sqrt{\lambda}$. To estimate $|S|$ within a $(1+\epsilon)$ factor by vertex sampling, it suffices to sample $\tilde{O}(n/|S|)$ vertices and decide their membership in $S$. A naive implementation costs $\tilde{O}(n^2/|S|)=\tilde{O}(n^{1+2/3})$ expected queries if $|S| \sim n^{1/3}$.

If $U$ were known explicitly, membership of a random vertex could be decided in expected $\tilde{O}(|U|)$ time: check whether it lies in $S_b$, outside $U$, or else run the MIS oracle restricted to $G[U]$. This would give total expected query complexity
$\tilde{O}(n|U|/|S|)$. Using $|U|=\tilde{O}(n/\lambda)$ and $|S|\geq |S_b|\geq n^{2/3}/\sqrt{\lambda}$, this is at most $\tilde{O}(n^{1+1/3}/\sqrt{\lambda})=\tilde{O}(n^{1+1/3})$.

The difficulty is that $U$ is not explicitly available. Verifying whether a vertex lies in $U$ may require $|S_b|$ queries, and a black-box invocation of the MIS oracle gives no improvement over previous bounds. In the hard regime $|S|\sim \sqrt n$ and $\lambda\sim n^{1/3}$, this leads back to the $\tilde{O}(n^{1+1/2})$ bound of Mahabadi et al.~\cite{MahabadiRTV25}.

Our main technical contribution is in showing that, using the partially computed set $S_b$, the membership of a vertex of $U$ in $S$ can nevertheless be decided in expected $\tilde{O}(n)$ queries; see \Cref{sec:query-complexity}. Consequently, the expected cost of deciding the status of a uniformly random vertex becomes $\tilde{O}(|S_b|+|U|)$: in $|S_b|$ time we test whether the vertex is in $S_b$, dominated by $S_b$, or potentially in $U$, and with probability $|U|/n$ we spend $\tilde{O}(n)$ additional queries. Therefore the total sampling cost is
$\tilde{O}(n(|S_b|+|U|)/|S|)=\tilde{O}(n+n|U|/|S|)=\tilde{O}(n^{1+1/3})$.

The intuition behind the fact that we may efficiently query MIS membership status even when restricted to $U$ is as follows. Suppose we run the adjacency-matrix implementation of the MIS oracle of ~\cite{YoshidaYMIH09} on $G[U]$. Since the algorithm does not know $U$, it must simulate queries in the original graph and filter out vertices outside $U$. The key insight is that, for a fixed vertex $u\in U$, most neighbors of $u$ were already banned early by vertices of $S_b$. Using the standard fact that after the greedy algorithm has selected $x$ vertices, the remaining unbanned graph has maximum degree $\tilde{O}(n/x)$, we show that for a random vertex $v\in V$, either $(u,v)\notin E$, or it suffices in expectation to inspect only the first $\tilde{O}(1)$ relevant vertices of $S_b$ to determine whether $v\in U$.

The MIS membership oracle of ~\cite{YoshidaYMIH09} runs in expected $O(|U|)$ time when restricted to $G[U]$. If it is simulated on the whole graph $G$, then each query to a vertex of $U$ is preceded, in expectation, by about $n/|U|$ queries to vertices outside $U$. The observation above shows that these rejected vertices can be filtered using only $\tilde{O}(1)$ expected work each. Thus the overhead of simulating the oracle on $G$ rather than on $G[U]$ is only $\tilde{O}(n)$ in expectation, giving the desired query bound.

\section{Preliminaries}
Given a graph $G = (V, E)$, an \emph{independent set} is a subset of vertices $S \subseteq V$ such that $(u,v) \notin E$ for every pair $u,v \in S$. An independent set $S$ is called \emph{maximal} if $S \cup \lbrace v\rbrace$ is not independent for any $v \in V \setminus S$. Throughout the paper we will use the term banned instead of dominated as common in the MIS literature. 

The greedy independent set algorithm first fixes some permutation of the vertices $\pi : [n] \rightarrow V$ and starts with an empty independent set $S =\emptyset$. The algorithm iterates through the vertices according to the ordering defined by $\pi$ and adds $v \in V$ to $S$ if $v$ neighbors no vertices in $S$. A run of the greedy independent set algorithm is said to be random if $\pi$ is chosen uniformly at random. For a permutation of vertices $\pi$ let $\text{RGMIS}(G,\pi)$ stand for the output of the greedy independent set algorithm when run on graph $G$ with respect to $\pi$.

Throughout this paper, we use standard notation where $|V| = n$ and $|E| = m$. Define the maximum degree of graph $G$ as $\Delta(G)$. For a subset $U \subseteq V$, let $G[U]$ denote the vertex-induced subgraph of $G$ on $U$. For $v\in V$ let $N_G(v)$ stand for the neighbors of $v$ in $G$. Analogously, for $U \subseteq V$ define $N_G(U) = \cup_{v \in U} N_G(v)$. We consider the \emph{adjacency matrix query} model, in which the graph is accessed via queries of the form $(u,v)$, where $u,v \in V$, returning whether $(u,v) \in E$. Sometimes we will refer to the \emph{adjacency list model} where the algorithm can use queries of the form $(v,i) | v \in V, i \in [\text{deg}_G(v)]$ returning the $i$-th neighbor of $v$.

Define an estimate $\tilde{\mu}$ to be a $(1+\epsilon)$-approximation of $\mu$ if $\tilde{\mu}\leq \mu \leq \tilde{\mu} (1+\epsilon)$. Much of the analysis of our algorithm relies on a well-known version of the Chernoff bound for binomial random variables $B(n,p)$, stated below.

\begin{lemma}[Chernoff Bounds] \label{lem:chernoff}
For any $\epsilon \in \left[0, 1\right]$,
$$\prob{B(n,p)-np  >\epsilon n p} \le e^{\frac{-\eps^2 np}{3}},$$
and 
$$\prob{np-B(n,p)  >\epsilon n p} \le e^{\frac{-\eps^2 np}{2}}.$$
\end{lemma}

In this paper we consider the problem of estimating the size of a maximal independent set computed by the greedy algorithm. In order to achieve this we will often refer to two different sublinear algorithms for deciding the membership of a vertex in the output independent set. The following algorithm is the list-query greedy MIS oracle.

\begin{algorithm}[H]
	
    \caption{$\mathcal{O}_L(G = (V,E),\pi,v)$ \label{alg:list-query}}
    \For{$u \in N_G(v) \cap \pi([\pi^{-1}(v)]-1)$ according to the ordering $\pi$}{
        \If{$\mathcal{O}_L(G,\pi,u) = 1$}{
            \Return{0}
        }
    }
    \Return{1}
\end{algorithm}

\begin{theorem}[\cite{YoshidaYMIH09}]
\label{thm:list-query}

$\mathcal{O}_L$ (\Cref{alg:list-query}) when queried with graph $G = (V,E)$ and a permutation $\pi$ of $V$ and vertex $v \in V$ drawn uniformly at random returns $1$ if and only if $v \in \text{RGMIS}(G,\pi)$ in expected $O(\Delta(G))$ time.

\end{theorem}

The second algorithm naturally implements the same subroutine using adjacency matrix queries.

\begin{algorithm}[H]
	
    \caption{$\mathcal{O}_M(G = (V,E),\pi,v)$ \label{alg:matrix-query}}
    \For{$i \in [\pi^{-1}(v)-1]$}{
        \If{$(v,\pi(i)) \in E$}{
            \If{$\mathcal{O}_M(G,\pi,\pi(i)) = 1$}{
                \Return{0}
            }
        }
    }
    \Return{1}
\end{algorithm}

\begin{theorem}[\cite{MahabadiRTV25}]
\label{thm:matrix-query}

$\mathcal{O}_M$ (\Cref{alg:matrix-query}) when queried with graph $G = (V,E)$ and a permutation $\pi$ of $V$ and vertex $v \in V$ drawn uniformly at random returns $1$ if and only if $v \in \text{RGMIS}(G,\pi)$ in expected $O(n)$ time.

\end{theorem}
\section{Estimating $|\text{RGMIS}(G,\pi)|$}

When initiated on graph $G = (V,E)$ the algorithm first draws $\pi: [n] \rightarrow V$ a uniform random permutation of $V$. Afterwards, the algorithm runs the greedy MIS algorithm with respect to $\pi$ for $\Theta(n^{1+1/3})$ steps in which a vertex is added to the MIS if it doesn't neighbor prior MIS vertices. The independent set found this way serves as the \emph{baseline} independent set $S_b$.

If the algorithm has iterated through the whole node set during this process it may return the size of the MIS exactly. Otherwise, it samples $\tilde{O}(n/|S_b|)$ vertices at random. For each sampled vertex $v$ it checks if it's present in $\text{RGMIS}(G,\pi)$. If $v \in S_b \cup N_G(S_b)$ this can be done in $O(|S_b|)$ queries greedily. Otherwise, the algorithm calls $\mathcal{O}_{S_b}(G,\pi,v)$ (\Cref{alg:U-oracle}).

Restricted to $U = V \setminus (S_b \cup N_G(S_b))$ the oracle $\mathcal{O}_{S_b}$ is identical to the adjacency matrix MIS oracle of \Cref{alg:matrix-query}. While simulating the MIS oracle on $U$ the algorithm has to detect if relevant vertices fall in $V \setminus U$. This is done through greedily checking edges to $S_b$ in an order respecting $\pi$.

\begin{algorithm}[H]

    $\pi: [n] \rightarrow V$ random permutation of $V$

    $S_b \leftarrow \emptyset$
    
    $i \leftarrow 1$

    \While{total adjacency queries made $\le 2n^{1+1/3}$ } {
        \If{$i > n$}{
            \Return{$|S_b|$}
        }
        \ElseIf{$\pi(i) \notin N_G(S_b)$}
        {
            $S_b \leftarrow S_b \cup \left\{\pi\left(i\right)\right\}$ 
        }
        $i$ = $i + 1$
    }

    $i_0 \gets$ the last vertex considered in the baseline loop
    
    $r \leftarrow \frac{100n\log n}{\left|S_b\right|\varepsilon^2}$
    
    \For{$i \ = 1, \dots, r$}
    {
        $v \leftarrow$ random vertex of $V$ 
        
        \If{$v \in S_b$}
        {
            $X_i \leftarrow 1$\
        }
        \ElseIf{$v \in N_G(S_b)$}
        {
            $X_i \leftarrow 0$\
        }
        \Else
        {
            $X_i \leftarrow \mathcal{O}_{S_b}\left(\pi,v\right)$\
        }
    }
    \Return{$n \cdot \frac{\sum_{i = 1}^{r} X_i}{r(1+\epsilon/2)} $}\
    \caption{Estimate-MIS$(G,\epsilon)$}
	\label{alg:two-phase}
\end{algorithm}

\begin{algorithm}[H]
    \For{$i = i_0+1, \dots ,\pi^{-1} (v)$}
        {
            \If{$\left(v, \ \pi\left(i\right)\right) \in E$}
            {
                $j \leftarrow 1$

                \While{$(\pi(i),S_b(j)) \notin E$}{
                    $j \leftarrow j+1$
                    \If{$j > |S_b|$}{
                        \If{$\mathcal{O}_{S_b}\left(G,\pi,\pi(i)\right) = 1$}
                            {
                                \Return{0}
                            }
                    }
                }
            }
        }
    \Return{1}
    \caption{$\mathcal{O}_{S_b}\left(\pi,v\right)$}
    \label{alg:U-oracle}
\end{algorithm}

\begin{lemma} \label{lem:correct approximation}
    With high probability, \Cref{alg:two-phase} returns a $(1+\epsilon)$-approximation to $|\text{RGMIS}(G,\pi)|$.
\end{lemma}
\begin{proof}
    First note that as $S_b$ is an independent set obtained by partially running the greedy MIS algorithm with respect to $\pi$ it must be that $S_b \subseteq S = \text{RGMIS}(G,\pi)$. Let $\bar{X} = \sum_{i=1}^r X_i$. Observe that $\bar{X} = B(r,|\text{RGMIS}(G,\pi)|/n)$. The proof concludes from a standard application of Chernoff's bound (\Cref{lem:chernoff}).

    \begin{align*}
    \mathbb{P}\left[\frac{n\bar{X}}{r(1+\epsilon/2)} \notin [\frac{|S|}{1+\epsilon},|S|]\right] & \leq \mathbb{P} \left[\frac{n\bar{X}}{r} -|S|> \frac{\epsilon |S|}{2}\right] + \mathbb{P} \left[|S| - \frac{n\bar{X}}{2} > \frac{\epsilon/2 |S|}{1 + \epsilon}\right]\\
    & \leq \mathbb{P}\left[B(r,\frac{|S|}{n}) - \frac{r |S|}{n} > \frac{\epsilon}{2}\frac{rn}{|S|}\right] + \mathbb{P}\left[\frac{r |S|}{n} - B(r,\frac{|S|}{n}) > \frac{\epsilon/2}{1 + \epsilon}\frac{rn}{|S|}\right]\\
    & \leq e^{-\frac{100|S| \log n}{12|S_b|}}  \leq 2e^{-5\log n} + e^{-\frac{100|S| \log n}{8|S_b|(1 + \epsilon)^2}}  \leq 2e^{-3\log n} = \frac{1}{\poly(n)}
    \end{align*}
\end{proof}

The following lemma upper bounds the query complexity of \Cref{alg:U-oracle} when queried with a uniform random vertex of $U = V \setminus (S_b \cup N_G(S_b))$. Its proof is deferred to \Cref{sec:query-complexity}.

\begin{lemma} \label{lem:linear-oracle}
    Let $T_{\mathcal{O}_{S_b}} (G,\pi,\ v)$ be the number of adjacency queries $\mathcal{O}_{S_b}\left(G,\pi, \ v\right)$ makes. Then, if $v$ is a uniform random vertex of $V \setminus (S_b \cup N_G(S_b))$,
    $$\exp{\pi,v}{T_{\mathcal{O}_{S_b}} (G,\pi,v)} = O(n\log ^2 n).$$
\end{lemma}

During the construction of $S_b$ in a run of \Cref{alg:two-phase} let $\lambda$ stand for the length of the longest series of consecutive vertices considered and not added to $S_b$. 

\begin{observation}
\label{obs:lambdasb}

$$|S_b| \geq \frac{n^{2/3}}{\sqrt{\lambda+1}}$$

\end{observation}

\begin{proof}

Whenever \Cref{alg:two-phase} considers adding a vertex $v$ to $S_b$ it only completes adjacency matrix queries to check whether $v \in N_G(S_b)$. Therefore, every $(\lambda+1)|S_b|$ adjacency queries it made it must have extended $S_b$ by one vertex. This implies $|S_b| \geq 2n^{1+1/3}/((\lambda+1) |S_b|)$. 

\end{proof}

\begin{claim}
\label{cl:lambda-bound}

With probability $(1-1/n^{10})$ (with respect to the randomness of $\pi$),

$$|U| \leq \frac{100n\log n }{\lambda}$$

\end{claim}

\begin{proof}

Let $\pi_{\leq i}$ stand for the permutation $\pi$ restricted to the first $i$ positions. For the sake of contradiction, assume that $|U| > 100 n \log n/\lambda$. Let $\pi(i_0)$ be the last vertex the algorithm considered to add to $S_b$. There must exist some index $\hat{i} < i_0$ such that $\pi(\hat{i}+1),\dots, \pi(\hat{i} + \lambda)$ were not added to $S_b$. 

Once $\pi(\hat{i})$ has been processed by the algorithm there must have been at least $|U|$ vertices not in or neighboring $S_b$. For a specific index $\hat{i}$ this event has probability at most

$$\prod_{i = 0}^{\lambda-1}\frac{n-|U|-\hat{i} - i}{n-\hat{i}-i} \leq \prod_{i = 0}^{\lambda-1}(1 - \frac{|U|}{n-\hat{i}-i}) \leq e^{-\lambda|U|/n} \leq n^{-20}.$$

Union bounding across all $O(n)$ indices $\hat{i}$ could be located at, we still get that this event has probability at most $n^{-10}$.

\end{proof}

We are now ready to upper bound the expected running time of \Cref{alg:two-phase}.

\begin{lemma}
\label{lem:running-time}

\Cref{alg:two-phase} runs in expected $O(\frac{n^{1+1/3} \log^4 n}{\epsilon^2})$ time.

\end{lemma}
\begin{proof}

The algorithm runs in time proportional to its query complexity. By definition, during the construction of the baseline MIS $S_b$ the algorithm makes $O(n^{1+1/3})$ queries to the adjacency matrix. Afterwards, it samples $O(n\log n/(|S_b| \cdot \epsilon^2))$ vertices uniformly at random and decides their membership in $\text{RGMIS}(G,\pi)$.

If a sampled vertex falls in $S_b \cup N_G(S_b)$ then in $O(|S_b|)$ queries the algorithm decides its MIS status. As vertices are sampled independently (with replacement) this has probability $|S_b \cup N_G(S_b)|/n$ for each sampled vertex.

If a sampled vertex falls in $U = V \setminus (S_b \cup N_G(S_b))$ by \Cref{lem:linear-oracle} the algorithm completes $O(n \log^ 2(n))$ queries to the adjacency matrix to decide its MIS status in expectation. This happens with $|U|/n$ probability. Therefore, the expected query complexity of the algorithm is upper bounded by 

\begin{align*}
O(n^{1+1/3} + \frac{n \log n}{|S_b|\epsilon^2} (\frac{|S_b \cup N_G(S_b)|}{n}  |S_b| + \frac{|U|}{n}  n \log^2 n)) & = O(n^{1+1/3} + \frac{|S_b \cup N_G(S_b)| \log n}{\epsilon^2} + \frac{|U|n \log^3 n}{|S_b|\epsilon^2}) \\
& = O(\frac{n^{1+1/3} \log n}{\epsilon^2} + \frac{|U|n \log^3 n}{|S_b|\epsilon^2})
\end{align*}

The second equality holds as $S_b \cup N_G(S_b) \subseteq V$. By \Cref{obs:lambdasb} we have that $|S_b| \geq n^{2/3}/\sqrt{\lambda+1}$. By \Cref{cl:lambda-bound} with high probability $|U| \leq O(n \log n/ \lambda)$. If this is the case the last running time term works out to be upper bounded by $O(n^{1+1/3} \log^4 n/\epsilon^2)$ dominating the query complexity. 

Otherwise, with probability at most $1/n^{10}$, we do not have control over the size of $|U|$. In this case we can only claim that $|U|n/|S_b| \leq n^2$. This event only contributes $O(n^{-8} \log^4n/\epsilon^2)$ to the expected query complexity. Therefore, the expected running time of the algorithm is $O(n^{1+1/3} \log^4 n/\epsilon^2)$.

\end{proof}

\begin{remark}

Note that the running time of the algorithm can be simply improved by an $O(\log^2 n)$ factor by spending that much more time on the construction of $S_b$. In turn this would mean that $U$ ends up being smaller, balancing the runtime of the two components. We ignored this small improvement in order to simplify presentation.

\end{remark}
\section{Query Complexity of $\mathcal{O}_{S_b}$}
\label{sec:query-complexity}

This section is devoted to upper bounding the expected query complexity of \Cref{alg:U-oracle} when queried with a uniform random vertex of $U = V \setminus (S_b \cup N_G(S_b))$, that is to the proof of \Cref{lem:linear-oracle}. Define $\pi_{i_0}$ to be the suffix of permutation $\pi$ \textbf{after} the $i_0$-th vertex (the last vertex fully processed in the first loop). Also, define $\pi_{\le i_0}$ to be the prefix of $\pi$ up to $i_0$. Let $V_{i_0}$ stand for the vertices of $\pi_{i_0}$ and let $V_{\le i_0}$ to be the vertices in the first $i_0$ positions. Similarly, define $\pi_U$ and $\pi_{V_{i_0}\setminus U}$ to be the permutation $\pi$ restricted to the vertices of $U$ and $V_{i_0} \setminus U$ respectively. We will consider the following different types of queries made by the algorithm.

\begin{definition}[$U$-Query$\left(u, v\right)$] \label{query:v-u}
    Given vertices $u \in U$ and  $v \in V$ is $(u,v) \in E$ and if so, is $v$ banned by $S_b$?
\end{definition}

Whenever $\mathcal{O}_{S_b}(\pi, u)$ is called the algorithm iterates through a suffix of $\pi_{i_0}$ and for any $v$ contained in this suffix it completes a $U$-query between $u$ and $v$. Note that the algorithm implements a $U$-query by checking $(u,v) \in E$ first. If an edge exists, the algorithm checks whether $v$ is banned by $S_b$. We will later show that this greatly reduces the expected query complexity of such queries.

\begin{definition} [Real Matrix-Query]
    During a run of \Cref{alg:U-oracle} we define an adjacency query $(u,v)$ made by the algorithm to be a real matrix-query if $u,v \in U$.
\end{definition}

Observe that restricted to the adjacency matrix of $G[U]$ a call to $\mathcal{O}_{S_b}(\pi, u)$ makes the same adjacency queries as $\mathcal{O}_M(G[U],\pi_U,u)$. By \Cref{thm:matrix-query} we know that in expectation the algorithm makes $O(|U|)$ real matrix-queries to $G[U]$.

\begin{definition} [Real List-Query] 
    During a run of \Cref{alg:U-oracle} we define an adjacency query $(u,v)$ made by the algorithm to be a real list-query if $u,v \in U, (u,v) \in E$.
\end{definition}

Similarly, restricted to the adjacency lists of $G[U]$,  $\mathcal{O}_{S_b}(\pi, u)$ behaves identically to $\mathcal{O}_{L}(G[U],\pi,u)$. By \Cref{thm:list-query} we know that the expected number of real list-queries is $O(\Delta(G[U]))$. Note that every real list-query is also a real matrix-query.

\begin{definition} [Spurious Query]
    During a run of \Cref{alg:U-oracle} we define an adjacency query $(u,v)$ made by the algorithm to be a spurious query if $u \in U, v \notin U$.
\end{definition}

Define the query cost of a real query or spurious query between vertices $u,v$ to be the total number of adjacency queries the algorithm makes while checking if $(u,v) \in E$ and if $v \in N_G(S_b)$. Denote the total query costs of the two types of queries as $T^R_{\mathcal{O}_{S_b}}(\pi, u)$ and $T^S_{\mathcal{O}_{S_b}}(\pi, u)$ respectively. We are then interested in upper bounding

$$\exp{\pi, u}{T_{\mathcal{O}_{S_b}}(\pi, u)} = \exp{\pi, u}{T^R_{\mathcal{O}_{S_b}}(\pi, u)} +\exp{\pi,u}{T^S_{\mathcal{O}_{S_b}}(G,\pi,u)}.$$

A crucial tool in our analysis is a folklore lemma governing the maximum degree of the induced sub-graph of un-banned vertices during a run of random order greedy maximal independent set once $t$ vertices have been added to the output. Let $S_b(i)$ and $S_b(\leq i)$ stand for the $i$-th vertex added to $S_b$ and the set of the first $i$ vertices added to $S_b$ during the run of \cref{alg:two-phase}.

\begin{lemma} [Proposition 3.1 of \cite{BehnezhadDHSS19}] \label{lem:U-max-degree}
    For $t \in [n]$ define $U_t = V \setminus (S_b(\leq t) \cup N_G(S_b(\leq t)))$, that is the subset of vertices not banned by the time $S_b$ grew to size $t$. Then with probability at least $1-1/n^{10}$, it holds that,
    
    $$\forall t \in [n] : \ \Delta(G[U_t]) \leq O(\frac{n \log n}{t}).$$
\end{lemma} 

Note that a slightly stronger variant of \Cref{lem:U-max-degree} also holds where instead of focusing on when $S_b$ grew to size $t$, we consider the time the $t$-th vertex was considered to be added to $S_b$. For our application, in the tight instances almost every vertex considered to be added to $S_b$ is added to $S_b$ making this weaker statement sufficient.

Let $\Delta_{S_b}$ stand for the residual degree bound event of \Cref{lem:U-max-degree} holding for any $t \leq |S_b|$. Most of the following discussion will be conditioned on $\Delta_{S_b}$. Furthermore, assume throughout this section that we \textbf{always} condition on $\pi$-s contents until $\pi^{-1}(i_0)$. Observe that under this conditioning $\pi_{i_0}$ is a uniform random permutation of $V_{i_0}$ and conditioning on this prefix fixes the size of $U$ and $V_{i_0}$.

\begin{claim}

\label{cl:adjacency-list:total}

$\exp{\pi_{i_0},u}{T^R_{\mathcal{O}_{S_b}}(\pi,u) | \Delta_{S_b}} = O(n \log n)$

\end{claim}

\begin{proof}

By \Cref{thm:list-query} and \Cref{thm:matrix-query} the $\mathcal{O}_{S_b}(\pi, u)$ is expected to complete $O(|U|)$ real matrix-queries and $O(\Delta(G[U]))$ real list-queries. When the algorithm completes a real list-query between vertices $(u,v)$ it needs to verify that $v \notin N_G(S_b)$ and $(u,v) \in E$. This requires $|S_b| + 1$ adjacency queries to $G$. When the algorithm completes a real matrix-query that is not a real list-query between the same vertices it must find that $(u,v) \notin E$ which requires a single query to $G$-s adjacency matrix. This implies,

\begin{align}
\exp{\pi_{i_0},u}{T^R_{\mathcal{O}_{S_b}}(\pi,u) | \Delta_{S_b}} & = O((|S_b| +1) \Delta(G[U]) + |U|)  \nonumber \\
& = O((|S_b| + 1) \cdot \frac{n \log n}{|S_b|} + |U|) \label{eq:cl:real}\\
& = O(n \log n) \nonumber
\end{align}

Here \Cref{eq:cl:real} follows from \Cref{lem:U-max-degree}.

\end{proof}

We now turn to upper bounding $\exp{\pi_{i_0},u}{T^S_{\mathcal{O}_{S_b}}(\pi,u) | \Delta_{S_b}}$ which will prove to be more difficult. We first prove the following surprising claim.

\begin{claim}
\label{cl:U-query:bound}

Let $u$ be an arbitrary vertex of $U$ and $v$ be a random vertex of $V_{i_0} \setminus U$. $\mathcal{O}_{S_b}$ completes a $U$-query between $u$ and $v$ in expected $O(\frac{n \log^2 n}{|V_{i_0} \setminus U|})$ adjacency queries when conditioned on $\Delta_{S_b}$.

\end{claim}

\begin{proof}
Define $U_i=V\setminus(S_b(\le max\{2^i, |S_b|\})\cup N_G(S_b(\le max\{2^i, |S_b|\})))$ for
$i\in\{0,\ldots,L\}$  with $U_{L + 1} = U$, where $L=\lceil \log_2 |S_b|\rceil$. Condition on the event $\Delta_{S_b}$, which implies by \Cref{lem:U-max-degree} that for every $i$, $\Delta(G[U_i]) =  O(n\log n/2^i)$ and $\Delta(G[U]) = O( n \log n/|S_b|)$. By definition, $U_i \supseteq U_{i+1}$ for $i \in \{0,\dots,L-1\}$ and $u$ has at most $O(n\log n/2^i)$ neighbors in $U_i \setminus U_{i+1}$. 

Similarly, $u$ has at most $O(\frac{n \log n}{|S_b|})$ neighbors in $U$. Define the query cost of the $U$-query as $T_U(u,v)$. If $v \in U$ then $\mathcal{O}_{S_b}$ needs to query $v$-s complete neighborhood with $S_b$ and $(u,v)$ which takes $|S_b| + 1$ queries. If $v \in U_i \setminus U_{i+1}$ then to answer the same query $\mathcal{O}_{S_b}$ only needs to query $(u,v)$ and $v$-s neighborhood in $S_b(\le max\{2^i, |S_b|\})$ which takes at most $2^{i+1} +1$ queries. Therefore,

\begin{align*}
\exp{\pi_{i_0},v}{T_U(u,v) | \Delta_{S_b}} & \leq 1 + \sum_{i = 0}^{L} \prob{v}{v \in N_G(u) \cap (U_i \setminus U_{i+1}) | \Delta_{S_b}} \cdot 2^{i + 1} + \prob{v}{v \in N_G(u) \cap U | \Delta_{S_b}} \cdot |S_b|\\
& \leq 1 + \sum_{i = 0}^{L} O(\frac{n \log n}{2^i |V_{i_0} \setminus U|} \cdot 2^{i+1}) + O(\frac{n \log n}{|S_b| \cdot |V_{i_0} \setminus U|} \cdot |S_b|) \\
& = O(\frac{n \log^2 n}{|V_{i_0} \setminus U|}).
\end{align*}

\end{proof}

We are now ready to derive our upper bound $\exp{\pi_{i_0},u}{T^S _{\mathcal{O}_{S_b}}(\pi,u) | \Delta_{S_b}}$.

\begin{lemma}

\label{lem:spurious:upper}

$\exp{\pi_{i_0},u}{T^S _{\mathcal{O}_{S_b}}(\pi,u) | \Delta_{S_b}} = O(n \log^2 n)$

\end{lemma}

\begin{proof}

First consider the adjacency queries made by $\mathcal{O}_M(G[U],\pi_U,u)$. Note that  $\mathcal{O}_{S_b}(\pi, u)$ makes the same adjacency queries to $G[U]$. $\mathcal{O}_{S_b}(\pi,u)$ will be queried for a series of vertices $v_1,\dots v_k$. Here $v_1 = u$ and as we assume the algorithm does not rely on caching, some vertices might appear multiple times in this series. When evaluating vertex $v_i$ it will query $v_i$-s neighborhood to some prefix of $\pi_U$, until the $X_i$-th position. For any $w$ in this prefix such that $(v_i,w) \in E$ the oracle also calls $\mathcal{O}_M(G[U],\pi_U,w)$. Here we will overload notation and define $[X_i]$ as the set of indices of the first $X_i$ vertices in $\pi_U$, which $v_i$ queries for adjacency, mapped to the original permutation $\pi$. 

Let $\bar{X} = \sum_{i \leq k} X_i$. (Therefore, $\bar{X} \ge k - 1$.) As $\pi_U$ is drawn uniformly at random (even when conditioned on the prefix of $\pi$ to the $i_0$-th element and $\Delta_{S_b})$ by \Cref{thm:matrix-query} we know that $\exp{\pi_{i_0},u}{\bar{X}| \Delta_{S_b}} = O(|U|)$. $\mathcal{O}_{S_b}(\pi, u)$ will be queried with the same series of vertices. For each $v_i$ it will complete $U$-queries between $v_i$ and some prefix of $\pi_{i_0}$, until the $Y_i$-th element where $\pi_{i_0}(Y_i) = \pi_U(X_i)$. Define $\bar{Y} = \sum_{i \leq k} Y_i$. Here, $[Y_i]$ does not need extra clarification, as it refers to the first $Y_i$ vertices, which $v_i$ does a $U$-query for, in the original permutation $\pi$ itself.

\begin{claim}
\label{cl:spurious:Yexp}
$\exp{\pi}{\bar{Y} |\bar{X} = x, \Delta_{S_b}} \leq \frac{x(|V_{i_0}| + 1)}{|U| + 1}$
\end{claim}

\begin{proof}

When $\mathcal{O}_{S_b}(G,\pi,v_i)$ is called the algorithm will complete $U$-queries between $v_i$ and a prefix of $\pi_{i_0}$ until the $X_i$-th element of $U$. The distribution describing how many elements of $V_{i_0}$ appear before the $x$-th element of $U$ is a negative hyper-geometric distribution with parameters $|V_{i_0}|,|U|,x$. Its expectation is known to be:

$$\exp{\pi}{Y_i | X_i = x} = \frac{x(|V_{i_0}|+1)}{|U|+1} \leq \frac{x (|V_{i_0}| + 1)}{|U| + 1}$$

Then, by conditioning on the series $X_{v_1},\dots,X_{v_k}$,

\begin{align*}
\exp{\pi_{i_0},u}{\bar{Y} |\bar{X} = x, \Delta_{S_b}, X_i = x_i \text{ for } i \leq k} & = \sum_{i \leq k} \exp{\pi}{Y_i | X_i = x_i, \Delta_{S_b}} \\
& \leq \sum_{i \leq k} \frac{x_i (|V_{i_0}| + 1)}{|U| + 1} = \frac{x (|V_{i_0}| + 1)}{|U| + 1}.
\end{align*}

The claim follows by summing over all potential conditionings on $x_1,\dots,x_k$.

\end{proof}

During the run of $\mathcal{O}_{S_b}(\pi, u)$ when querying $v_i$ the algorithm will complete $Y_i - X_i \leq Y_i$ spurious queries as it will query all edges between $v_i$ and the first $Y_i - X_i$ vertices of $\pi_{i_0}$ which are not in $U$. Each of these will only require a $U$-query. Let $T_U(v_i,i)$ stand for the cost of a $U$-query between $(v_i,\pi_{i_0}(i))$. Then, conditioning on $\bar{Y},\bar{X},\Delta_{S_b}$ the total query cost of all the spurious queries the algorithm makes can be upper bounded by:

\begin{align}
    &\exp{\pi_{i_0}, u}{T_{\mathcal{O}_{S_b}}^S(\pi, u) \mid \Delta_{S_b}, \bar{Y}=y, \bar{X} = x} \nonumber \\
    &\quad = \exp{u}{\exp{\pi_U}{\begin{aligned}[t]
        &\exp{\pi_{V_{i_0} \setminus U}}{\sum_{i \leq k} \sum_{j \in [Y_i] \setminus [X_i]} T_U(v_i,\pi_{i_0}(j))} \\
        &\quad \mid \Delta_{S_b}, \bar{Y}=y, \bar{X} = x \end{aligned}}} \nonumber \\
    &\quad = \exp{u}{\exp{\pi_U}{\begin{aligned}[t]
        &\sum_{i \leq k} \sum_{j \in [Y_i] \setminus [X_i]} \exp{\pi_{V_{i_0} \setminus U}}{T_U(v_i,\pi_{i_0}(j))} \\
        &\quad \mid \Delta_{S_b}, \bar{Y}=y, \bar{X} = x \end{aligned}}} \nonumber \\
    &\quad = \exp{u}{\exp{\pi_U}{\begin{aligned}[t]
        &\sum_{i \leq k} \sum_{j \in [Y_i] \setminus [X_i]} \exp{\pi_{V_{i_0} \setminus U}}{T_U(v_i,\pi_{i_0}(j)) \mid \Delta_{S_b}} \\
        &\quad \mid \bar{Y}=y, \bar{X} = x \end{aligned}}} \nonumber \\
    &\quad = \exp{u}{\exp{\pi_U}{\begin{aligned}[t]
        &\sum_{i \leq k} \sum_{j \in [Y_i] \setminus [X_i]} \frac{n \log ^2 n}{|V_{i_0} \setminus U|}\\
        &\quad \mid \bar{Y}=y, \bar{X} = x \end{aligned}}} \label{eq:spurious:1}\\
    &\quad \leq O\left(\frac{n (y - x) \log^2 n}{|V_{i_0} \setminus U|}\right) \nonumber.
\end{align}

Here, all vertices $v_i$ are in $U$ and $\pi_{i_0}(j)$ is a uniform random vertex of $V_{i_0} \setminus U$ for any $j \in [Y_i] \setminus [X_i]$. Furthermore, . Therefore, we can invoke \Cref{cl:U-query:bound} to get equation \Cref{eq:spurious:1}. In the following derivation assume throughout that we are conditioning on $\Delta_{S_b}$ which we will omit from the notation for the sake of convenience.

\begin{align}
\exp{\pi_{i_0},u}{T_{\mathcal{O}_{S_b}}^S(\pi, u)} & = \sum_{x}\prob{\pi,u}{\bar{X} = x} \exp{\pi_{i_0},u}{T_{\mathcal{O}_{S_b}}^S(\pi, u)| \bar{X} = x} \nonumber \\
& = \sum_{x}\prob{\pi,u}{\bar{X} = x} \sum_{y} \prob{\pi,u}{\bar{Y} = y \mid \bar{X} = x}\exp{\pi_{i_0},u}{T_{\mathcal{O}_{S_b}}^S(\pi, u) |  \bar{X} = x, \bar{Y}= y} \nonumber \\
& \leq \sum_{x}\prob{\pi,u}{\bar{X} = x} \sum_{y} \prob{\pi,u}{\bar{Y} = y \mid \bar{X} = x} \cdot (y - x) \cdot O(\frac{n \log ^2 n}{|V_{i_0} \setminus U|}) \label{eq:spurious:3} \\
& = \sum_{x}\prob{\pi,u}{\bar{X} = x} \left(\exp{\pi_{i_0},u}{\bar{Y} | \bar{X} = x } - x\right) \cdot O(\frac{n \log ^2 n}{|V_{i_0} \setminus U|}) \nonumber\\
& \leq \sum_{x}\prob{\pi,u}{\bar{X} = x} \cdot \left(O(\frac{x(|V_{i_0}| + 1)}{|U| + 1}) - x\right) \cdot  O(\frac{n \log ^2 n}{|V_{i_0}| \setminus U}) \label{eq:spurious:4}\\
& = \exp{\pi_{i_0},u}{\bar{X}} \cdot O(\frac{n \log^2 n}{|U|}) \nonumber \\
& = O(n \log^2 n) \label{eq:spurious:5}
\end{align}

\Cref{eq:spurious:3}, \Cref{eq:spurious:4} and \Cref{eq:spurious:5} follow from \Cref{eq:spurious:1}, \Cref{cl:spurious:Yexp} and \Cref{thm:matrix-query} respectively.

\end{proof}

If we are satisfied with conditioning on the high probability event $\Delta_{S_b}$ we can already conclude \Cref{lem:linear-oracle}. To make sure that the expected query bound holds in general we have to upper bound $\exp{\pi_{i_0},u}{T_{\mathcal{O}_{S_b}}^S(\pi, u) | \neg \Delta_{S_b}}$.

\begin{claim}
\label{cl:U-oracle:deterministic}

$$\exp{\pi_{i_0},u}{T_{\mathcal{O}_{S_b}}(\pi, u) | \neg \Delta_{S_b}} = O(n^3)$$

\end{claim}

\begin{proof}

Even if $\Delta_{S_b}$ does not hold, the permutation $\pi_U$ is drawn uniformly at random over $U$. Hence, by \Cref{thm:matrix-query} we have that 
$$\exp{\pi_{i_0},u}{T^R _{\mathcal{O}_{S_b}}(G[U],\pi,u) | \neg \Delta_{S_b}} = \exp{\pi_{i_0},u}{T_{\mathcal{O}_M}(G[U],\pi,u) | \neg \Delta_{S_b}} = O(|U|) = O(n).$$
$\mathcal{O}_{S_b}(\pi, u)$ completes the queries of $\mathcal{O}_M(G[U],\pi,u)$. Between any two queries made by the latter oracle, the former oracle might complete $U$-queries between the queried node and the whole of $V_{i_0}$. A $U$-query is executed through at most $|S_b|+1$ adjacency queries. Therefore,

\begin{align*}
\exp{\pi_{i_0},u}{T_{\mathcal{O}_{S_b}}^S(\pi, u) | \neg \Delta_{S_b}} & \leq \exp{\pi_{i_0},u}{T_{M}(G[U],\pi_U,u) | \neg \Delta_{S_b}} \cdot |V_{i_0}| \cdot (|S_b| + 1) \\
&  = O(n^3)
\end{align*}

\end{proof}

We are now ready to formally bound $\exp{\pi_{i_0},u}{T_{\mathcal{O}_{S_b}}(\pi, u)}$.

\begin{align*}
\exp{\pi_{i_0},u}{T_{\mathcal{O}_{S_b}}(\pi, u)} & = \prob{\pi}{\Delta_{S_b}}\exp{\pi_{i_0},u}{T_{\mathcal{O}_{S_b}}(\pi, u)|\Delta_{S_b}} + \prob{\pi}{\neg\Delta_{S_b}}\exp{\pi_{i_0},u}{T_{\mathcal{O}_{S_b}}(\pi, u)|\neg \Delta_{S_b}} \\
& \leq \exp{\pi_{i_0},u}{T_{\mathcal{O}_{S_b}}(\pi, u)|\Delta_{S_b}} + O(\frac{1}{n^4}) \\
& = \exp{\pi_{i_0},u}{T^S_{\mathcal{O}_{S_b}}(\pi, u) + T^R_{\mathcal{O}_{S_b}}(\pi, u)|\Delta_{S_b}} + O(\frac{1}{n^4}) \\
& = O(n \log^2 n)
\end{align*}

Note that throughout these derivations we have conditioned on the prefix of $\pi$ up till the $i_0$-th element. As none of the arguments depend on the contents of the actual prefix (just the event $\Delta_{S_b}$) we may sum across all potential conditioning of the prefix of $\pi$.

\section{Sublinear Metric $k$-Center via Independent Set}

\label{sec:k-center}

\newcommand{\costk}{\operatorname{cost}}
\newcommand{\OPTk}{\operatorname{OPT}}
\newcommand{\MISk}{\operatorname{MIS}}

In the metric $k$-center problem, we are given a set $P$ of $n$ points in a metric space. We may query pairwise distances between points. The goal is to find $C \subseteq P$, $|C| \leq k$, minimizing $\costk(C) = \max_{p \in P} d(p,C)$. Let $\OPTk_k(P)$ denote the cost of an optimum solution. Let $\Delta$ denote the aspect ratio of $P$, that is an upper bound on the ratio of the maximum and minimum pairwise distance. Most authors assume $\Delta = \poly(n)$.

An $(\alpha,\beta)$-bi-criteria estimate $\tilde{r}$ for the cost of the $k$-center objective is a value satisfying

$$\OPTk_{\alpha k}(P) \leq \tilde{r}\leq \beta \OPTk_{k}(P).$$

In this section, we show that the folklore $2$-approximation algorithm for metric $k$-center (see \cite{kcenterbase}) can be naturally extended to the sublinear model (where sub-linearity is meant in terms of $nk$) for $k \in \omega(n^{1/3})$, through sublinear independent-set estimation.

\begin{lemma}
\label{lem}

Assume there exists a sublinear algorithm which, given adjacency-matrix query access to a graph $G$, returns a $(1+\epsilon)$-approximation to the size of a maximal independent set of $G$ in time $T(\epsilon)$.

Then there exists a sublinear algorithm returning a $(1+\epsilon,2+\epsilon)$-bi-criteria approximation to the metric $k$-center objective on $n$ points in time $O(T(\frac{\epsilon}{10}) \log(\Delta)/\epsilon)$.

\end{lemma}

\begin{proof}

For a given radius $r$, define the \emph{threshold graph} $G_r = (P,E_r)$, where there is an edge between points $p,p' \in P$ in $E_r$ if and only if $d(p,p') \leq r$. Let $\MISk(G_r)$ denote an arbitrary maximal independent set of $G_r$.

\begin{claim}
\label{cl}

The following two implications hold.

\begin{itemize}
\item If $|\MISk(G_r)| \leq k$ then $\OPTk_k(P) \leq r$.
\item If $\OPTk_k(P) \leq r/2$ then $|\MISk(G_r)| \leq k$.
\end{itemize}

\end{claim}

\begin{proof}

To see the first part of the claim, let $C=\MISk(G_r)$. Since $C$ is a maximal independent set of $G_r$, every point $p \in P \setminus C$ has a neighbor in $C$. Therefore, $d(p,C) \leq r$ for every point $p \in P$. Hence, if $|C|\leq k$, then $C$ is a feasible $k$-center solution of cost at most $r$, and so $\OPTk_k(P)\leq r$.

For the second part, assume $\OPTk_k(P) \leq r/2$, and let $C$ be an optimal solution to the metric $k$-center problem on $P$. Suppose, toward a contradiction, that there exists an independent set $S$ of $G_r$ of size greater than $k$. Since $|S|>|C|$, there must be two distinct points $p,p' \in S$ whose closest center in $C$ is the same point $c \in C$. But then $d(p,p') \leq d(p,c) + d(c,p') \leq r$, meaning that $(p,p') \in E_r$, contradicting the independence of $S$. Thus every independent set of $G_r$ has size at most $k$, and in particular $|\MISk(G_r)| \leq k$.

\end{proof}

We now describe the algorithm. Let $\eta=\epsilon/10$. The algorithm selects an arbitrary pair of distinct points $p,p' \in P$ and sets $r_0 = d(p,p')/\Delta$. By the definition of the aspect ratio, every nonzero pairwise distance in $P$ is at least $r_0$. It then considers the geometrically increasing set of radii $r_i = (1+\eta)^i \cdot r_0$ for $i \in [O(\log(\Delta)/\eta)]$, large enough so that the final tested radius is at least twice the diameter of $P$. For each radius $r_i$, the algorithm implicitly defines the threshold graph $G_{r_i}$ and runs the assumed MIS-size estimator on this graph to obtain an estimate $\widetilde{s}(r_i)$ satisfying
$$
\widetilde{s}(r_i) \leq |\MISk(G_{r_i})| \leq (1+\eta)\widetilde{s}(r_i).
$$
The algorithm returns the smallest radius $r_i$ satisfying
$$
\widetilde{s}(r_i) \leq (1+\eta)k.
$$

We prove correctness. First, consider any radius $r$ such that $\OPTk_k(P)\leq r/2$. By \cref{cl}, $|\MISk(G_r)|\leq k$. Since the estimator is an $\eta$-approximation in the sense above, we have
$$
\widetilde{s}(r)\leq |\MISk(G_r)|\leq k.
$$
Therefore $\widetilde{s}(r)\leq (1+\eta)k$, and so such a radius is accepted. Since the grid is multiplicative, there is a tested radius $r_i$ satisfying
$$
2\OPTk_k(P)\leq r_i\leq 2(1+\eta)\OPTk_k(P).
$$
This radius is accepted, and hence the returned radius $\tilde{r}$ satisfies
$$
\tilde{r}\leq 2(1+\eta)\OPTk_k(P).
$$

It remains to show the bi-criteria covering guarantee. Let $\tilde{r}$ be the radius returned by the algorithm. Since $\widetilde{s}(\tilde{r})\leq (1+\eta)k$ and the estimator is an $\eta$-approximation, we have
$$
|\MISk(G_{\tilde{r}})| \leq (1+\eta)\widetilde{s}(\tilde{r}) \leq (1+\eta)^2 k \leq (1+\epsilon)k.
$$
By maximality of $\MISk(G_{\tilde{r}})$, every point of $P$ is within distance at most $\tilde{r}$ of some vertex in $\MISk(G_{\tilde{r}})$. Hence $P$ can be covered by at most $(1+\epsilon)k$ centers of radius $\tilde{r}$, and therefore
$
\OPTk_{(1+\epsilon)k}(P) \leq \tilde{r}.
$

Combining the two inequalities gives
$$
\OPTk_{(1+\epsilon)k}(P) \leq \tilde{r}\leq 2(1+\eta)\OPTk_k(P).
$$
Since $\eta=\epsilon/10$, we have $2(1+\eta)\leq 2+\epsilon$, and after rescaling constants this gives the claimed bi-criteria guarantee. The number of tested radii is $O(\log(\Delta)/\epsilon)$, and each invocation of the MIS-size estimator costs $T(\frac{\epsilon}{10})$ time. Therefore the total running time is $O(T(\frac{\epsilon}{10}) \log(\Delta)/\epsilon)$.

\end{proof}

\paragraph{Acknowledgment.} We thank Sayan Bhattacharya for several helpful discussions.

Arash Kooroshnezhad is funded by the European Union (ERC grant, DYNALP, 101170133). Views and opinions expressed are however those of the author(s) only and do not necessarily reflect those of the European Union or the European Research Council Executive Agency. Neither the European Union nor the granting authority can be held responsible for them.




\bibliographystyle{alpha}
\bibliography{bibfile}

@inproceedings{BehnezhadDHSS19,
  author       = {Soheil Behnezhad and
                  Mahsa Derakhshan and
                  MohammadTaghi Hajiaghayi and
                  Cliff Stein and
                  Madhu Sudan},
  editor       = {David Zuckerman},
  title        = {Fully Dynamic Maximal Independent Set with Polylogarithmic Update
                  Time},
  booktitle    = {60th {IEEE} Annual Symposium on Foundations of Computer Science, {FOCS}
                  2019, Baltimore, Maryland, USA, November 9-12, 2019},
  pages        = {382--405},
  publisher    = {{IEEE} Computer Society},
  year         = {2019},
  url          = {https://doi.org/10.1109/FOCS.2019.00032},
  doi          = {10.1109/FOCS.2019.00032},
  bibsource    = {dblp computer science bibliography, https://dblp.org}
}

@inproceedings{YoshidaYMIH09,
author = {Yoshida, Yuichi and Yamamoto, Masaki and Ito, Hiro},
title = {An improved constant-time approximation algorithm for maximum~matchings},
year = {2009},
isbn = {9781605585062},
publisher = {Association for Computing Machinery},
address = {New York, NY, USA},
url = {https://doi.org/10.1145/1536414.1536447},
doi = {10.1145/1536414.1536447},
booktitle = {Proceedings of the Forty-First Annual ACM Symposium on Theory of Computing},
pages = {225–234},
numpages = {10},
location = {Bethesda, MD, USA},
series = {STOC '09}
}

@inproceedings{KapralovMNT20,
  author       = {Michael Kapralov and
                  Slobodan Mitrovic and
                  Ashkan Norouzi{-}Fard and
                  Jakab Tardos},
  editor       = {Shuchi Chawla},
  title        = {Space Efficient Approximation to Maximum Matching Size from Uniform
                  Edge Samples},
  booktitle    = {Proceedings of the 2020 {ACM-SIAM} Symposium on Discrete Algorithms,
                  {SODA} 2020, Salt Lake City, UT, USA, January 5-8, 2020},
  pages        = {1753--1772},
  publisher    = {{SIAM}},
  year         = {2020},
  url          = {https://doi.org/10.1137/1.9781611975994.107},
  doi          = {10.1137/1.9781611975994.107},
  bibsource    = {dblp computer science bibliography, https://dblp.org}
}

@inproceedings{Behnezhad21,
  author       = {Soheil Behnezhad},
  title        = {Time-Optimal Sublinear Algorithms for Matching and Vertex Cover},
  booktitle    = {62nd {IEEE} Annual Symposium on Foundations of Computer Science, {FOCS}
                  2021, Denver, CO, USA, February 7-10, 2022},
  pages        = {873--884},
  publisher    = {{IEEE}},
  year         = {2021},
  url          = {https://doi.org/10.1109/FOCS52979.2021.00089},
  doi          = {10.1109/FOCS52979.2021.00089},
  bibsource    = {dblp computer science bibliography, https://dblp.org}
}

@inproceedings{BehnezhadRR23b,
  author       = {Soheil Behnezhad and
                  Mohammad Roghani and
                  Aviad Rubinstein},
  editor       = {Barna Saha and
                  Rocco A. Servedio},
  title        = {Sublinear Time Algorithms and Complexity of Approximate Maximum Matching},
  booktitle    = {Proceedings of the 55th Annual {ACM} Symposium on Theory of Computing,
                  {STOC} 2023, Orlando, FL, USA, June 20-23, 2023},
  pages        = {267--280},
  publisher    = {{ACM}},
  year         = {2023},
  url          = {https://doi.org/10.1145/3564246.3585231},
  doi          = {10.1145/3564246.3585231},
  bibsource    = {dblp computer science bibliography, https://dblp.org}
}

@inproceedings{BehnezhadRRS23,
  author       = {Soheil Behnezhad and
                  Mohammad Roghani and
                  Aviad Rubinstein and
                  Amin Saberi},
  editor       = {Nikhil Bansal and
                  Viswanath Nagarajan},
  title        = {Beating Greedy Matching in Sublinear Time},
  booktitle    = {Proceedings of the 2023 {ACM-SIAM} Symposium on Discrete Algorithms,
                  {SODA} 2023, Florence, Italy, January 22-25, 2023},
  pages        = {3900--3945},
  publisher    = {{SIAM}},
  year         = {2023},
  url          = {https://doi.org/10.1137/1.9781611977554.ch151},
  doi          = {10.1137/1.9781611977554.CH151},
  bibsource    = {dblp computer science bibliography, https://dblp.org}
}

@inproceedings{BhattacharyaKS23a,
  author       = {Sayan Bhattacharya and
                  Peter Kiss and
                  Thatchaphol Saranurak},
  title        = {Dynamic (1+$\epsilon$)-Approximate Matching Size in Truly Sublinear
                  Update Time},
  booktitle    = {64th {IEEE} Annual Symposium on Foundations of Computer Science, {FOCS}
                  2023, Santa Cruz, CA, USA, November 6-9, 2023},
  pages        = {1563--1588},
  publisher    = {{IEEE}},
  year         = {2023},
  url          = {https://doi.org/10.1109/FOCS57990.2023.00095},
  doi          = {10.1109/FOCS57990.2023.00095},
  bibsource    = {dblp computer science bibliography, https://dblp.org}
}

@inproceedings{AzarmehrBRR25,
  author       = {Amir Azarmehr and
                  Soheil Behnezhad and
                  Mohammad Roghani and
                  Aviad Rubinstein},
  title        = {Tight Pair Query Lower Bounds for Matching and Earth Mover's
                  Distance},
  booktitle    = {66th {IEEE} Annual Symposium on Foundations of Computer Science, {FOCS}
                  2025, Sydney, Australia, December 14-17, 2025},
  pages        = {2666--2693},
  publisher    = {{IEEE}},
  year         = {2025},
  url          = {https://doi.org/10.1109/FOCS63196.2025.00138},
  doi          = {10.1109/FOCS63196.2025.00138},
  bibsource    = {dblp computer science bibliography, https://dblp.org}
}

@inproceedings{MahabadiRT25,
  author       = {Sepideh Mahabadi and
                  Mohammad Roghani and
                  Jakub Tarnawski},
  editor       = {Keren Censor{-}Hillel and
                  Fabrizio Grandoni and
                  Jo{\"{e}}l Ouaknine and
                  Gabriele Puppis},
  title        = {A 0.51-Approximation of Maximum Matching in Sublinear n{\^{}}\{1.5\}
                  Time},
  booktitle    = {52nd International Colloquium on Automata, Languages, and Programming,
                  {ICALP} 2025, Aarhus, Denmark, July 8-11, 2025},
  series       = {LIPIcs},
  volume       = {334},
  pages        = {116:1--116:17},
  publisher    = {Schloss Dagstuhl - Leibniz-Zentrum f{\"{u}}r Informatik},
  year         = {2025},
  url          = {https://doi.org/10.4230/LIPIcs.ICALP.2025.116},
  doi          = {10.4230/LIPICS.ICALP.2025.116},
  bibsource    = {dblp computer science bibliography, https://dblp.org}
}

@inproceedings{BehnezhadRR24,
  author       = {Soheil Behnezhad and
                  Mohammad Roghani and
                  Aviad Rubinstein},
  editor       = {Bojan Mohar and
                  Igor Shinkar and
                  Ryan O'Donnell},
  title        = {Approximating Maximum Matching Requires Almost Quadratic Time},
  booktitle    = {Proceedings of the 56th Annual {ACM} Symposium on Theory of Computing,
                  {STOC} 2024, Vancouver, BC, Canada, June 24-28, 2024},
  pages        = {444--454},
  publisher    = {{ACM}},
  year         = {2024},
  url          = {https://doi.org/10.1145/3618260.3649785},
  doi          = {10.1145/3618260.3649785},
  bibsource    = {dblp computer science bibliography, https://dblp.org}
}

@inproceedings{BehnezhadRR23a,
  author       = {Soheil Behnezhad and
                  Mohammad Roghani and
                  Aviad Rubinstein},
  title        = {Local Computation Algorithms for Maximum Matching: New Lower Bounds},
  booktitle    = {64th {IEEE} Annual Symposium on Foundations of Computer Science, {FOCS}
                  2023, Santa Cruz, CA, USA, November 6-9, 2023},
  pages        = {2322--2335},
  publisher    = {{IEEE}},
  year         = {2023},
  url          = {https://doi.org/10.1109/FOCS57990.2023.00143},
  doi          = {10.1109/FOCS57990.2023.00143},
  bibsource    = {dblp computer science bibliography, https://dblp.org}
}

@inproceedings{BhattacharyaKS23b,
  author       = {Sayan Bhattacharya and
                  Peter Kiss and
                  Thatchaphol Saranurak},
  editor       = {Barna Saha and
                  Rocco A. Servedio},
  title        = {Sublinear Algorithms for (1.5+{\(\epsilon\)})-Approximate Matching},
  booktitle    = {Proceedings of the 55th Annual {ACM} Symposium on Theory of Computing,
                  {STOC} 2023, Orlando, FL, USA, June 20-23, 2023},
  pages        = {254--266},
  publisher    = {{ACM}},
  year         = {2023},
  url          = {https://doi.org/10.1145/3564246.3585252},
  doi          = {10.1145/3564246.3585252},
  bibsource    = {dblp computer science bibliography, https://dblp.org}
}

@inproceedings{ChenKK20,
  author       = {Yu Chen and
                  Sampath Kannan and
                  Sanjeev Khanna},
  editor       = {Artur Czumaj and
                  Anuj Dawar and
                  Emanuela Merelli},
  title        = {Sublinear Algorithms and Lower Bounds for Metric {TSP} Cost Estimation},
  booktitle    = {47th International Colloquium on Automata, Languages, and Programming,
                  {ICALP} 2020, Saarbr{\"{u}}cken, Germany (Virtual Conference),
                  July 8-11, 2020},
  series       = {LIPIcs},
  volume       = {168},
  pages        = {30:1--30:19},
  publisher    = {Schloss Dagstuhl - Leibniz-Zentrum f{\"{u}}r Informatik},
  year         = {2020},
  url          = {https://doi.org/10.4230/LIPIcs.ICALP.2020.30},
  doi          = {10.4230/LIPICS.ICALP.2020.30},
  bibsource    = {dblp computer science bibliography, https://dblp.org}
}

@inproceedings{0039KT23b,
  author       = {Yu Chen and
                  Sanjeev Khanna and
                  Zihan Tan},
  editor       = {Kousha Etessami and
                  Uriel Feige and
                  Gabriele Puppis},
  title        = {Sublinear Algorithms and Lower Bounds for Estimating {MST} and {TSP}
                  Cost in General Metrics},
  booktitle    = {50th International Colloquium on Automata, Languages, and Programming,
                  {ICALP} 2023, Paderborn, Germany, July 10-14, 2023},
  series       = {LIPIcs},
  volume       = {261},
  pages        = {37:1--37:16},
  publisher    = {Schloss Dagstuhl - Leibniz-Zentrum f{\"{u}}r Informatik},
  year         = {2023},
  url          = {https://doi.org/10.4230/LIPIcs.ICALP.2023.37},
  doi          = {10.4230/LIPICS.ICALP.2023.37},
  bibsource    = {dblp computer science bibliography, https://dblp.org}
}

@inproceedings{BehnezhadRRS24,
  author       = {Soheil Behnezhad and
                  Mohammad Roghani and
                  Aviad Rubinstein and
                  Amin Saberi},
  editor       = {Karl Bringmann and
                  Martin Grohe and
                  Gabriele Puppis and
                  Ola Svensson},
  title        = {Sublinear Algorithms for {TSP} via Path Covers},
  booktitle    = {51st International Colloquium on Automata, Languages, and Programming,
                  {ICALP} 2024, Tallinn, Estonia, July 8-12, 2024},
  series       = {LIPIcs},
  volume       = {297},
  pages        = {19:1--19:16},
  publisher    = {Schloss Dagstuhl - Leibniz-Zentrum f{\"{u}}r Informatik},
  year         = {2024},
  url          = {https://doi.org/10.4230/LIPIcs.ICALP.2024.19},
  doi          = {10.4230/LIPICS.ICALP.2024.19},
  bibsource    = {dblp computer science bibliography, https://dblp.org}
}

@article{CzumajS09,
  author       = {Artur Czumaj and
                  Christian Sohler},
  title        = {Estimating the Weight of Metric Minimum Spanning Trees in Sublinear
                  Time},
  journal      = {{SIAM} J. Comput.},
  volume       = {39},
  number       = {3},
  pages        = {904--922},
  year         = {2009},
  url          = {https://doi.org/10.1137/060672121},
  doi          = {10.1137/060672121},
  bibsource    = {dblp computer science bibliography, https://dblp.org}
}

@article{ChazelleRT05,
  author       = {Bernard Chazelle and
                  Ronitt Rubinfeld and
                  Luca Trevisan},
  title        = {Approximating the Minimum Spanning Tree Weight in Sublinear Time},
  journal      = {{SIAM} J. Comput.},
  volume       = {34},
  number       = {6},
  pages        = {1370--1379},
  year         = {2005},
  url          = {https://doi.org/10.1137/S0097539702403244},
  doi          = {10.1137/S0097539702403244},
  bibsource    = {dblp computer science bibliography, https://dblp.org}
}

@inproceedings{LenzenL18,
  author       = {Christoph Lenzen and
                  Reut Levi},
  editor       = {Ioannis Chatzigiannakis and
                  Christos Kaklamanis and
                  D{\'{a}}niel Marx and
                  Donald Sannella},
  title        = {A Centralized Local Algorithm for the Sparse Spanning Graph Problem},
  booktitle    = {45th International Colloquium on Automata, Languages, and Programming,
                  {ICALP} 2018, Prague, Czech Republic, July 9-13, 2018},
  series       = {LIPIcs},
  volume       = {107},
  pages        = {87:1--87:14},
  publisher    = {Schloss Dagstuhl - Leibniz-Zentrum f{\"{u}}r Informatik},
  year         = {2018},
  url          = {https://doi.org/10.4230/LIPIcs.ICALP.2018.87},
  doi          = {10.4230/LIPICS.ICALP.2018.87},
  bibsource    = {dblp computer science bibliography, https://dblp.org}
}

@inproceedings{ArvivCLP23,
  author       = {Rubi Arviv and
                  Lily Chung and
                  Reut Levi and
                  Edward Pyne},
  editor       = {Nicole Megow and
                  Adam D. Smith},
  title        = {Improved Local Computation Algorithms for Constructing Spanners},
  booktitle    = {Approximation, Randomization, and Combinatorial Optimization. Algorithms
                  and Techniques, {APPROX/RANDOM} 2023, Atlanta, Georgia, USA, September
                  11-13, 2023},
  series       = {LIPIcs},
  volume       = {275},
  pages        = {42:1--42:23},
  publisher    = {Schloss Dagstuhl - Leibniz-Zentrum f{\"{u}}r Informatik},
  year         = {2023},
  url          = {https://doi.org/10.4230/LIPIcs.APPROX/RANDOM.2023.42},
  doi          = {10.4230/LIPICS.APPROX/RANDOM.2023.42},
  bibsource    = {dblp computer science bibliography, https://dblp.org}
}

@inproceedings{ParterRVY19,
  author       = {Merav Parter and
                  Ronitt Rubinfeld and
                  Ali Vakilian and
                  Anak Yodpinyanee},
  editor       = {Avrim Blum},
  title        = {Local Computation Algorithms for Spanners},
  booktitle    = {10th Innovations in Theoretical Computer Science Conference, {ITCS}
                  2019, San Diego, California, USA, January 10-12, 2019},
  series       = {LIPIcs},
  volume       = {124},
  pages        = {58:1--58:21},
  publisher    = {Schloss Dagstuhl - Leibniz-Zentrum f{\"{u}}r Informatik},
  year         = {2019},
  url          = {https://doi.org/10.4230/LIPIcs.ITCS.2019.58},
  doi          = {10.4230/LIPICS.ITCS.2019.58},
  bibsource    = {dblp computer science bibliography, https://dblp.org}
}

@inproceedings{0039KT23,
  author       = {Yu Chen and
                  Sanjeev Khanna and
                  Zihan Tan},
  editor       = {Nikhil Bansal and
                  Viswanath Nagarajan},
  title        = {Query Complexity of the Metric Steiner Tree Problem},
  booktitle    = {Proceedings of the 2023 {ACM-SIAM} Symposium on Discrete Algorithms,
                  {SODA} 2023, Florence, Italy, January 22-25, 2023},
  pages        = {4893--4935},
  publisher    = {{SIAM}},
  year         = {2023},
  url          = {https://doi.org/10.1137/1.9781611977554.ch179},
  doi          = {10.1137/1.9781611977554.CH179},
  bibsource    = {dblp computer science bibliography, https://dblp.org}
}

@inproceedings{MahabadiRTV25b,
  author       = {Sepideh Mahabadi and
                  Mohammad Roghani and
                  Jakub Tarnawski and
                  Ali Vakilian},
  editor       = {Raghu Meka},
  title        = {Sublinear Metric Steiner Tree via Improved Bounds for Set Cover},
  booktitle    = {16th Innovations in Theoretical Computer Science Conference, {ITCS}
                  2025, Columbia University, New York, NY, USA, January 7-10, 2025},
  series       = {LIPIcs},
  volume       = {325},
  pages        = {74:1--74:24},
  publisher    = {Schloss Dagstuhl - Leibniz-Zentrum f{\"{u}}r Informatik},
  year         = {2025},
  url          = {https://doi.org/10.4230/LIPIcs.ITCS.2025.74},
  doi          = {10.4230/LIPICS.ITCS.2025.74},
  bibsource    = {dblp computer science bibliography, https://dblp.org}
}

@inproceedings{NguyenO08,
  author       = {Huy N. Nguyen and
                  Krzysztof Onak},
  title        = {Constant-Time Approximation Algorithms via Local Improvements},
  booktitle    = {49th Annual {IEEE} Symposium on Foundations of Computer Science, {FOCS}
                  2008, Philadelphia, PA, USA, October 25-28, 2008},
  pages        = {327--336},
  publisher    = {{IEEE} Computer Society},
  year         = {2008},
  url          = {https://doi.org/10.1109/FOCS.2008.81},
  doi          = {10.1109/FOCS.2008.81},
  bibsource    = {dblp computer science bibliography, https://dblp.org}
}

@inproceedings{000122,
  author       = {Mohsen Ghaffari},
  title        = {Local Computation of Maximal Independent Set},
  booktitle    = {63rd {IEEE} Annual Symposium on Foundations of Computer Science, {FOCS}
                  2022, Denver, CO, USA, October 31 - November 3, 2022},
  pages        = {438--449},
  publisher    = {{IEEE}},
  year         = {2022},
  url          = {https://doi.org/10.1109/FOCS54457.2022.00049},
  doi          = {10.1109/FOCS54457.2022.00049},
  bibsource    = {dblp computer science bibliography, https://dblp.org}
}

@inproceedings{RubinfeldTVX11,
  author       = {Ronitt Rubinfeld and
                  Gil Tamir and
                  Shai Vardi and
                  Ning Xie},
  editor       = {Bernard Chazelle},
  title        = {Fast Local Computation Algorithms},
  booktitle    = {Innovations in Computer Science - {ICS} 2011, Tsinghua University,
                  Beijing, China, January 7-9, 2011. Proceedings},
  pages        = {223--238},
  publisher    = {Tsinghua University Press},
  year         = {2011},
  url          = {http://conference.iiis.tsinghua.edu.cn/ICS2011/content/papers/36.html},
  bibsource    = {dblp computer science bibliography, https://dblp.org}
}

@inproceedings{AlonRVX12,
  author       = {Noga Alon and
                  Ronitt Rubinfeld and
                  Shai Vardi and
                  Ning Xie},
  editor       = {Yuval Rabani},
  title        = {Space-efficient local computation algorithms},
  booktitle    = {Proceedings of the Twenty-Third Annual {ACM-SIAM} Symposium on Discrete
                  Algorithms, {SODA} 2012, Kyoto, Japan, January 17-19, 2012},
  pages        = {1132--1139},
  publisher    = {{SIAM}},
  year         = {2012},
  url          = {https://doi.org/10.1137/1.9781611973099.89},
  doi          = {10.1137/1.9781611973099.89},
  bibsource    = {dblp computer science bibliography, https://dblp.org}
}

@article{ReingoldV16,
  author       = {Omer Reingold and
                  Shai Vardi},
  title        = {New techniques and tighter bounds for local computation algorithms},
  journal      = {J. Comput. Syst. Sci.},
  volume       = {82},
  number       = {7},
  pages        = {1180--1200},
  year         = {2016},
  url          = {https://doi.org/10.1016/j.jcss.2016.05.007},
  doi          = {10.1016/J.JCSS.2016.05.007},
  bibsource    = {dblp computer science bibliography, https://dblp.org}
}

@article{LeviRY17,
  author       = {Reut Levi and
                  Ronitt Rubinfeld and
                  Anak Yodpinyanee},
  title        = {Local Computation Algorithms for Graphs of Non-constant Degrees},
  journal      = {Algorithmica},
  volume       = {77},
  number       = {4},
  pages        = {971--994},
  year         = {2017},
  url          = {https://doi.org/10.1007/s00453-016-0126-y},
  doi          = {10.1007/S00453-016-0126-Y},
  bibsource    = {dblp computer science bibliography, https://dblp.org}
}

@inproceedings{GhaffariU19,
  author       = {Mohsen Ghaffari and
                  Jara Uitto},
  editor       = {Timothy M. Chan},
  title        = {Sparsifying Distributed Algorithms with Ramifications in Massively
                  Parallel Computation and Centralized Local Computation},
  booktitle    = {Proceedings of the Thirtieth Annual {ACM-SIAM} Symposium on Discrete
                  Algorithms, {SODA} 2019, San Diego, California, USA, January 6-9,
                  2019},
  pages        = {1636--1653},
  publisher    = {{SIAM}},
  year         = {2019},
  url          = {https://doi.org/10.1137/1.9781611975482.99},
  doi          = {10.1137/1.9781611975482.99},
  bibsource    = {dblp computer science bibliography, https://dblp.org}
}

@article{EvenMR18,
  author       = {Guy Even and
                  Moti Medina and
                  Dana Ron},
  title        = {Best of two local models: Centralized local and distributed local
                  algorithms},
  journal      = {Inf. Comput.},
  volume       = {262},
  pages        = {69--89},
  year         = {2018},
  url          = {https://doi.org/10.1016/j.ic.2018.07.001},
  doi          = {10.1016/J.IC.2018.07.001},
  bibsource    = {dblp computer science bibliography, https://dblp.org}
}

@article{LeviM17,
  author       = {Reut Levi and
                  Moti Medina},
  title        = {A (Centralized) Local Guide},
  journal      = {Bull. {EATCS}},
  volume       = {122},
  year         = {2017},
  url          = {http://eatcs.org/beatcs/index.php/beatcs/article/view/495},
  bibsource    = {dblp computer science bibliography, https://dblp.org}
}

@article{kcenterbase,
author = {Hochbaum, Dorit S. and Shmoys, David B.},
title = {A unified approach to approximation algorithms for bottleneck problems},
year = {1986},
issue_date = {July 1986},
publisher = {Association for Computing Machinery},
address = {New York, NY, USA},
volume = {33},
number = {3},
issn = {0004-5411},
url = {https://doi.org/10.1145/5925.5933},
doi = {10.1145/5925.5933},
journal = {J. ACM},
month = may,
pages = {533–550},
numpages = {18}
}

@inproceedings{BateniEFHJMW23,
  author       = {MohammadHossein Bateni and
                  Hossein Esfandiari and
                  Hendrik Fichtenberger and
                  Monika Henzinger and
                  Rajesh Jayaram and
                  Vahab Mirrokni and
                  Andreas Wiese},
  editor       = {Nikhil Bansal and
                  Viswanath Nagarajan},
  title        = {Optimal Fully Dynamic \emph{k}-Center Clustering for Adaptive and
                  Oblivious Adversaries},
  booktitle    = {Proceedings of the 2023 {ACM-SIAM} Symposium on Discrete Algorithms,
                  {SODA} 2023, Florence, Italy, January 22-25, 2023},
  pages        = {2677--2727},
  publisher    = {{SIAM}},
  year         = {2023},
  url          = {https://doi.org/10.1137/1.9781611977554.ch101},
  doi          = {10.1137/1.9781611977554.CH101},
  bibsource    = {dblp computer science bibliography, https://dblp.org}
}

@article{abs-2112-07050,
  author       = {MohammadHossein Bateni and
                  Hossein Esfandiari and
                  Rajesh Jayaram and
                  Vahab S. Mirrokni},
  title        = {Optimal Fully Dynamic k-Centers Clustering},
  journal      = {CoRR},
  volume       = {abs/2112.07050},
  year         = {2021},
  url          = {https://arxiv.org/abs/2112.07050},
  eprinttype   = {arXiv},
  eprint       = {2112.07050},
  bibsource    = {dblp computer science bibliography, https://dblp.org}
}

@article{Cohen-AddadFS21,
  author       = {Vincent Cohen{-}Addad and
                  Andreas Emil Feldmann and
                  David Saulpic},
  title        = {Near-linear Time Approximation Schemes for Clustering in Doubling
                  Metrics},
  journal      = {J. {ACM}},
  volume       = {68},
  number       = {6},
  pages        = {44:1--44:34},
  year         = {2021},
  url          = {https://doi.org/10.1145/3477541},
  doi          = {10.1145/3477541},
  bibsource    = {dblp computer science bibliography, https://dblp.org}
}

@inproceedings{PanconesiS92,
  author       = {Alessandro Panconesi and
                  Aravind Srinivasan},
  editor       = {S. Rao Kosaraju and
                  Mike Fellows and
                  Avi Wigderson and
                  John A. Ellis},
  title        = {Improved Distributed Algorithms for Coloring and Network Decomposition
                  Problems},
  booktitle    = {Proceedings of the 24th Annual {ACM} Symposium on Theory of Computing,
                  May 4-6, 1992, Victoria, British Columbia, Canada},
  pages        = {581--592},
  publisher    = {{ACM}},
  year         = {1992},
  url          = {https://doi.org/10.1145/129712.129769},
  doi          = {10.1145/129712.129769},
  bibsource    = {dblp computer science bibliography, https://dblp.org}
}

@article{BarenboimEPS16,
  author       = {Leonid Barenboim and
                  Michael Elkin and
                  Seth Pettie and
                  Johannes Schneider},
  title        = {The Locality of Distributed Symmetry Breaking},
  journal      = {J. {ACM}},
  volume       = {63},
  number       = {3},
  pages        = {20:1--20:45},
  year         = {2016},
  url          = {https://doi.org/10.1145/2903137},
  doi          = {10.1145/2903137},
  bibsource    = {dblp computer science bibliography, https://dblp.org}
}

@inproceedings{AwerbuchGLP89,
  author       = {Baruch Awerbuch and
                  Andrew V. Goldberg and
                  Michael Luby and
                  Serge A. Plotkin},
  title        = {Network Decomposition and Locality in Distributed Computation},
  booktitle    = {30th Annual Symposium on Foundations of Computer Science, Research
                  Triangle Park, North Carolina, USA, 30 October - 1 November 1989},
  pages        = {364--369},
  publisher    = {{IEEE} Computer Society},
  year         = {1989},
  url          = {https://doi.org/10.1109/SFCS.1989.63504},
  doi          = {10.1109/SFCS.1989.63504},
  bibsource    = {dblp computer science bibliography, https://dblp.org}
}

@inproceedings{RozhonG20,
  author       = {V{\'{a}}clav Rozhon and
                  Mohsen Ghaffari},
  editor       = {Konstantin Makarychev and
                  Yury Makarychev and
                  Madhur Tulsiani and
                  Gautam Kamath and
                  Julia Chuzhoy},
  title        = {Polylogarithmic-time deterministic network decomposition and distributed
                  derandomization},
  booktitle    = {Proceedings of the 52nd Annual {ACM} {SIGACT} Symposium on Theory
                  of Computing, {STOC} 2020, Chicago, IL, USA, June 22-26, 2020},
  pages        = {350--363},
  publisher    = {{ACM}},
  year         = {2020},
  url          = {https://doi.org/10.1145/3357713.3384298},
  doi          = {10.1145/3357713.3384298},
  bibsource    = {dblp computer science bibliography, https://dblp.org}
}

@article{BalliuBHORS21,
  author       = {Alkida Balliu and
                  Sebastian Brandt and
                  Juho Hirvonen and
                  Dennis Olivetti and
                  Mika{\"{e}}l Rabie and
                  Jukka Suomela},
  title        = {Lower Bounds for Maximal Matchings and Maximal Independent Sets},
  journal      = {J. {ACM}},
  volume       = {68},
  number       = {5},
  pages        = {39:1--39:30},
  year         = {2021},
  url          = {https://doi.org/10.1145/3461458},
  doi          = {10.1145/3461458},
  bibsource    = {dblp computer science bibliography, https://dblp.org}
}

@inproceedings{KothapalliP12,
  author       = {Kishore Kothapalli and
                  Sriram V. Pemmaraju},
  editor       = {Deepak D'Souza and
                  Telikepalli Kavitha and
                  Jaikumar Radhakrishnan},
  title        = {Super-Fast 3-Ruling Sets},
  booktitle    = {{IARCS} Annual Conference on Foundations of Software Technology and
                  Theoretical Computer Science, {FSTTCS} 2012, Hyderabad, India, December
                  15-17, 2012},
  series       = {LIPIcs},
  volume       = {18},
  pages        = {136--147},
  publisher    = {Schloss Dagstuhl - Leibniz-Zentrum f{\"{u}}r Informatik},
  year         = {2012},
  url          = {https://doi.org/10.4230/LIPIcs.FSTTCS.2012.136},
  doi          = {10.4230/LIPICS.FSTTCS.2012.136},
  bibsource    = {dblp computer science bibliography, https://dblp.org}
}

@article{BaswanaS07,
  author       = {Surender Baswana and
                  Sandeep Sen},
  title        = {A simple and linear time randomized algorithm for computing sparse
                  spanners in weighted graphs},
  journal      = {Random Struct. Algorithms},
  volume       = {30},
  number       = {4},
  pages        = {532--563},
  year         = {2007},
  url          = {https://doi.org/10.1002/rsa.20130},
  doi          = {10.1002/RSA.20130},
  bibsource    = {dblp computer science bibliography, https://dblp.org}
}

@article{AhmedBSHJKS20,
  author       = {Abu Reyan Ahmed and
                  Greg Bodwin and
                  Faryad Darabi Sahneh and
                  Keaton Hamm and
                  Mohammad Javad Latifi Jebelli and
                  Stephen G. Kobourov and
                  Richard Spence},
  title        = {Graph spanners: {A} tutorial review},
  journal      = {Comput. Sci. Rev.},
  volume       = {37},
  pages        = {100253},
  year         = {2020},
  url          = {https://doi.org/10.1016/j.cosrev.2020.100253},
  doi          = {10.1016/J.COSREV.2020.100253},
  bibsource    = {dblp computer science bibliography, https://dblp.org}
}

@article{JainV01,
  author       = {Kamal Jain and
                  Vijay V. Vazirani},
  title        = {Approximation algorithms for metric facility location and \emph{k}-Median
                  problems using the primal-dual schema and Lagrangian relaxation},
  journal      = {J. {ACM}},
  volume       = {48},
  number       = {2},
  pages        = {274--296},
  year         = {2001},
  url          = {https://doi.org/10.1145/375827.375845},
  doi          = {10.1145/375827.375845},
  bibsource    = {dblp computer science bibliography, https://dblp.org}
}

@inproceedings{MahabadiRTV25,
  author       = {Sepideh Mahabadi and
                  Mohammad Roghani and
                  Jakub Tarnawski and
                  Ali Vakilian},
  editor       = {Kasper Green Larsen and
                  Barna Saha},
  title        = {Sublinear Metric Steiner Forest via Maximal Independent Set},
  booktitle    = {Proceedings of the 2026 Annual {ACM-SIAM} Symposium on Discrete Algorithms,
                  {SODA} 2026, Vancouver, BC, Canada, January 11-14, 2026},
  pages        = {5941--5959},
  publisher    = {{SIAM}},
  year         = {2026},
  url          = {https://doi.org/10.1137/1.9781611978971.211},
  doi          = {10.1137/1.9781611978971.211},
  bibsource    = {dblp computer science bibliography, https://dblp.org}
}

@misc{mahabadi2026improvedsublinearalgorithmsmaximal,
      title={Improved Sublinear Algorithms for Maximal Independent Set and Metric Steiner Forest}, 
      author={Sepideh Mahabadi and Jakub Tarnawski},
      year={2026},
      eprint={2610.06053},
      archivePrefix={arXiv},
      primaryClass={cs.DS},
      url={https://arxiv.org/abs/2610.06053}, 
}


\end{document}